\documentclass[12pt]{article}  
\pdfoutput=1
\usepackage{mathptmx}
\usepackage{amsthm}
\usepackage{amssymb}
\usepackage{graphicx}
\newdimen\einr
\def\abs#1{\par\hangafter=1\hangindent=\einr
  \noindent\hbox to\einr{\ignorespaces#1\hfill}\ignorespaces} 

\newtheorem{theorem}{Theorem}[section]

\newtheorem{lemma}[theorem]{Lemma}
\newtheorem{proposition}[theorem]{Proposition}
\theoremstyle{definition}
\newtheorem{remark}[theorem]{Remark}
\newtheorem{definition}[theorem]{Definition}
\def\reals{{\mathbb R}} 
\def\eps{\varepsilon}
\def\prob{\hbox{prob}}

\def\mod{\hbox{mod}}
\def\proof{\noindent{\em Proof.\enspace}}

\def\endproof{\hfill\strut\nobreak\hfill\tombstone\par\medbreak}
\def\tombstone{\hbox{\lower.4pt\vbox{\hrule\hbox{\vrule
  \kern7.6pt\vrule height7.6pt}\hrule}\kern.5pt}}

\title{Nash Codes for Noisy Channels%
\thanks{We thank Drew Fudenberg for the suggestion of an
``ex ante'' proof of Theorem~\ref{t-global}, Rann
Smorodinsky for raising the question of potential functions
(see Proposition~\ref{p-pot}), Graham Brightwell for
a comment that led to the improved example in
Figure~\ref{binary}, and Christina Pawlowitsch and Joel
Sobel for stimulating discussions.
Three anonymous referees gave detailed suggestions that
improved this article significantly.  
General thanks go to Amparo Urbano and Jos\'e E. Vila
for continued support.
This work has been supported by the Spanish Ministry of
Science and Technology under project ECO2010-20584/ECON and
FEDER, PROMETEO/2009/068.
}
}

\author{Pen\'elope Hern\'andez%
\thanks{Department of Economic Analysis and ERI-CES,
University of Valencia,
46022 Valencia, Spain.
Email:  penelope.hernandez@uv.es
} \and
\and Bernhard von Stengel%
\thanks{Department of Mathematics,
London School of Economics, London WC2A 2AE, United Kingdom.
Email: stengel@nash.lse.ac.uk}
}

\date{February 16, 2014}

\begin{document}
\maketitle
\begin{abstract}
\noindent
This paper studies the stability of communication protocols
that deal with transmission errors.  
We consider a coordination game between an informed sender
and an uninformed decision maker, the receiver, who
communicate over a noisy channel.
The sender's strategy, called a code, maps states of nature
to signals.
The receiver's best response is to decode the received
channel output as the state with highest expected receiver
payoff.
Given this decoding, an equilibrium or ``Nash code'' results
if the sender encodes every state as prescribed.
We show two theorems that give sufficient conditions for
Nash codes.
First, a receiver-optimal code defines a Nash code.
A second, more surprising observation holds for
communication over a binary channel which is used
independently a number of times, a basic model of
information transmission: 
Under a minimal ``monotonicity'' requirement for breaking
ties when decoding, which holds generically, {\em every}
code is a Nash code.

%
\strut






\end{abstract}


\section{Introduction}

Information transmission is central to the interaction of
economic agents and to the operation of organizations.
This paper presents a game-theoretic analysis of
communication with errors over a ``noisy channel''.
The noisy channel is a basic model of information theory,
pioneered by Shannon (1948), and fundamental for the design
of reliable data transmission.
In this model, an informed sender sends a message, which is
distorted by the channel, to an uninformed receiver.
Sender and receiver have the common interest that the
receiver understands the sender as reliably as possible.

A communication protocol defines a code, that is, a set of
channel inputs that represent the possible messages for the
sender, and a way for the receiver to decode the channel
output.
One can view the designer of the protocol as a ``social
planner'' who tries to solve an optimization problem, for
example to achieve high reliability and a good rate of
information transmission.
This assumes that sender and receiver adhere to the
protocol.
In this paper, we study this model as a {\em game} between
sender and receiver as two players.
A strategy of the sender is a code, and a strategy of
the receiver is a way to decode the channel output.
Rather than requiring that sender and receiver adhere to
their respective strategies, we assume that they can choose
their strategies freely.
A {\em Nash equilibrium} is a pair of strategies for sender and
receiver that are mutual {\em best responses}.
This is the central stability concept of game theory.

The best response of the receiver is known in the
communications literature as MAP (maximum a posteriori)
decoding.
In contrast, allowing the sender to deviate from the code
(while the receiver strategy is fixed) is specific to the
game-theoretic approach.
If the sender is not in equilibrium, she has an incentive to
change her strategy to encode some message with a
different codeword.
If this happens, the protocol will lose its function as a
de-facto standard of communication.
The appeal of a Nash equilibrium is that it is
self-enforcing.

Sender-receiver games have attracted significant interest in
economics (Spence, 1973; Crawford and Sobel, 1982).
The game-theoretic view is also applied in models of
language evolution (Nowak and Krakauer, 1999; Argiento et
al., 2009).
These assume, as in our case, that the interests of sender
and receiver are fully aligned, and use Nash equilibrium as
the natural stability criterion.
We survey this related literature in more detail below.
In the analysis and design of communication networks, a
growing body of research deals with game-theoretic
approaches that assume selfish agents (Srivastava et al.,
2005; MacKenzie and DaSilva, 2006; Anshelevich et al.,
2008), again with Nash equilibrium as the central solution
concept.


\subsection*{The model}

We consider the classic model of the {\em discrete noisy
channel}.
The channel has a finite set of input and output symbols and
known transition probabilities that represent the possible
communication successes and errors.
The channel may also be used repeatedly, with independent
errors.
In the important case of the binary channel that has only
two symbols, the codewords are then fixed-length sequences
of bits.

In our sender-receiver game, nature chooses one of finitely
many states at random according to a prior probability.
The sender is informed of the state and transmits a signal via
the discrete noisy channel to the uninformed receiver who
makes a decision.
The sender's strategy or {\em code} assigns to each state of
nature a specific signal or ``codeword'' that is the input
to the channel.
The receiver's strategy decodes the distorted signal that is
the channel output as one of the possible states.
Both players receive a (possibly different) positive payoff
only if the state is decoded correctly, otherwise payoff
zero.

In equilibrium, the receiver decodes the channel output as
the state with highest expected payoff.
When all states get equal receiver payoff, the receiver
condition is the well-known MAP 
decoding rule (MacKay, 2003, p.~305).
The equilibrium condition for the sender means that she
chooses for each state the prescribed codeword as her best
response, that is, no other channel input has a higher
probability of being decoded correctly with the given
receiver strategy.

A {\em Nash code} is a code together with a best-response 
decoding function that defines a Nash equilibrium.
So we assume the straightforward equilibrium condition for
the receiver and require that the code fulfills the more
involved sender condition.
(Of course, both conditions are necessary for equilibrium.)

\subsection*{Our results}

We present two main results about Nash codes, along with
other observations that we describe in the outline of our
paper at the end of this introduction.
Our first main result concerns discrete channels with
arbitrary finite sets of input and output symbols.
We show that already for three symbols, not every code
defines a Nash equilibrium.
However, a Nash code results if the expected payoff to the
receiver cannot be increased by replacing a single codeword
with another one (Theorem~\ref{t-opt}).
So these {\em receiver-optimal} codes are Nash codes.
This is closely related to {\em potential games} 
(Proposition~\ref{p-pot}), which may provide the starting
point for studying dynamics of codes until they become Nash
codes, as a topic for further research.

In short, without any constraints on the channel, and for
any best-response decoding, receiver-optimal codes are Nash
codes.
For equal receiver utilities for each state, these are the
codes with maximum expected reliability, which therefore
implies Nash equilibrium.  
The method to show this result is not deep; its purpose is
to analyze our model.
The key assumption is that an improvement in decoding
probability benefits both sender and receiver.
However, a {\em sender-optimal} code is {\em not}
necessarily a Nash code if sender and receiver give
different utilities to a correct decoding of the state of
nature.  
This happens if the sender can use an unused message to
transmit the information about the state more reliably.
If all channel symbols are used, then under reasonable
assumptions the code is Nash (see
Proposition~\ref{p-diag}).\footnote{ We thank an anonymous
referee for suggesting this result.}

Our second main result is more surprising and technically
challenging.
It applies to the {\em binary channel} where codewords are
strings of bits with independent positive error
probabilities for each bit.
Then {\em every} code is a Nash code (Theorem~\ref{t-bin}),
irrespective of its quality.
The only requirement for the decoding is that the receiver
breaks ties between states {\em monotonically}, that is,
in a consistent manner; 
this holds for natural tie-breaking rules, and ties do not
even occur if states of nature have different generic prior
probabilities or utilities.  
That is, for the binary channel, as long as the receiver
decodes optimally and breaks ties consistently, the
equilibrium condition holds automatically on the sender's
side.

Binary codes are fundamental to the practice and theory of
information theory.
Our result that they are Nash codes shows that they are
incentive compatible.
Hence, this condition is orthogonal to engineering issues
such as high reliability and rate of information
transmission.


\subsection*{Related literature}

Information transmission is often modeled in the economic
literature as a sender-receiver game between an informed
expert and an uninformed decision maker.
Standard signaling models (pioneered by Spence, 1973) often
assume that signals have costs associated with the
information of the sender.
In their seminal work on strategic information transmission,
Crawford and Sobel (1982) consider costless signals and
communication without transmission errors, but where the
incentives of sender and receiver differ.
They assume that a fixed interval represents the set of
possible states, messages, and receiver's actions.
Payoffs depend continuously on the difference between state
and action, and differ for sender and receiver.
In equilibrium, the interval is partitioned into finitely
many intervals, and the sender sends as her message only the
partition class that contains the state.
Thus, the sender only reveals partial information about the
state.
Along with many other models (see the surveys by Kreps and
Sobel, 1994, and Sobel, 2013), this shows that information
is not transmitted faithfully for {\em strategic} reasons
because of some conflict of interest.

Even in rather simple sender-receiver games, players can get
higher equilibrium payoffs when communicating over a channel
with noise than with perfect communication (Myerson, 1994,
Section~4).
Blume, Board, and Kawamura (2007) extend the model by
Crawford and Sobel (1982) by assuming communication errors.
The noise allows for equilibria that improve welfare
compared to the Crawford--Sobel model.
The construction partly depends on the specific form of
the errors so that erroneous transmissions can be
identified; this does not apply in our discrete model.
In addition, in our model players only get positive payoff
when the receiver decodes the state correctly, unlike in the
continuous models by Crawford and Sobel (1982) and Blume et
al.\ (2007).
On the other hand, compared to perfect communication, noise 
may prevent players from achieving common knowledge about
the state of nature (Koessler, 2001).

Game-theoretic models of communication have been used in
the study of language (see De Jaegher and van Rooij, 2013,
for a recent survey).
Lewis (1969) describes language as a ``convention'' with 
mappings between states and signals, and argues that these
should be bijections.
Nowak and Krakauer (1999) use evolutionary game theory to
show how languages may evolve from ``noisy'' mappings;
W\"arneryd (1993) shows that only bijections are
evolutionary stable.
However, even ambiguous sender mappings (where one signal is
used for more than one state) together with a mixed receiver
population may be ``neutrally stable'' (Pawlowitsch, 2008);
the randomized receiver strategy can be seen as noise.
Argiento et al. (2009) consider the learning process of a
language in a sender-receiver game.
This is extended to the noisy channel by Touri and Lambort
(2013).

Blume and Board (2013) use the noisy channel to model
vagueness in communication.
Lipman (2009) discusses how vagueness can arise even for
coinciding interests of sender and receiver.
Ambiguous signals arise when the set of messages is smaller
than the set of states, which may reflect communication costs
for the sender (see J\"ager, Koch-Metzger, and Riedel, 2011,
and the discussion in Sobel 2012).
For the sender-receiver game with a noisy binary channel,
Hern\'andez, Urbano, and Vila (2012) describe the equilibria
for a specific code that can serve as a ``universal
grammar''; the explicit receiver strategy allows to
characterize the equilibrium payoff.

Noise in communication is relevant to models of persuasion,
where the sender wants to induce the receiver to take an
action.
Glazer and Rubinstein (2004; 2006) study binary receiver
actions; the sender may reveal limited information about the
state of nature as ``evidence''.
The optimal way to do so is a receiver-optimal mechanism.  
In a more general setting, Kamenica and Gentzkow (2011)
allow the sender to commit to a strategy that selects a
message for each state, assuming the receiver's best
response using Bayesian updating; the sender may generate
noise by selecting the message at random.
Subject to a certain Bayesian consistency requirement, the 
sender can commit to her best possible strategy.

Equilibrium models of information transmission give several
insights.
First, communication may fail:
Every sender-receiver game has a ``babbling equilibrium''
where the sender's action is independent of the state and
the receiver's action is independent of the channel output,
with no information transmitted.
Second, equilibria are typically not unique (for example,
mapping states to signals is often arbitrary).
Third, conflict of interest, or cost and complexity of
communication (Sobel, 2012), prevent perfect communication.

Our approach takes a basic view that communication can be
impeded by noise when interests of sender and receiver are
aligned, and analyzes this issue game-theoretically.
Our results show that the Nash equilibrium condition is
weaker than or, for the binary channel, orthogonal to the
quality of information transmission.


\subsection*{Outline of the paper }

Section~\ref{s-model} describes our model and characterizes
the Nash equilibrium condition.
For channels with any number of symbols,
Section~\ref{s-examples} gives examples that some codes may
not be Nash codes.
Section~\ref{s-recopt} shows that receiver-optimal codes are
Nash, and discusses the relation to potential functions.  
In Section~\ref{s-bin}, we consider binary codes, where we
first demonstrate that tie-breaking needs to be
``monotonic'' when ties occur in order for Nash equilibrium
to hold for every code.
In Section~\ref{s-insy} we show the main Theorem~\ref{t-bin}
that every monotonically decoded binary code is Nash.
This holds in fact not just for binary codes but for any
``input symmetric'' channels with any number of symbols
where the probability of receiving a symbol incorrectly does
not depend on the channel input.
The proof also shows that the property of a channel that
every code is Nash, which we call ``Nash-stability'',
extends to any product of channels (see
Section~\ref{s-stable}) with independent errors.
The product channel assumes independent error probabilities,
but the codewords are still arbitrary combinations of inputs
for such products.
(If the error probabilities are not independent, then the
channel has to be considered with $n$-tuples as input and
output symbols where in general only Theorem~\ref{t-opt}
about receiver-optimal codes applies.)
A natural monotonic decoding rule is to break ties according
to a fixed order among the states, as when they have generic
priors.  
In Section~\ref{s-mon} it is shown that this is in fact the
only general deterministic monotonic tie-breaking rule.  

\section{Nash codes}
\label{s-model}

We consider a game of two players, a sender (she) and a receiver (he). 
First, nature chooses a {\em state\/} $i$ from a set
$\Omega=\{0,1,\ldots,M-1\}$ with positive {\em prior\/}
probability~$q_i$.
Then the sender is fully informed about~$i$, and
sends a message to the receiver via a noisy channel.
After receiving the message as output by the channel, the
receiver takes an action that affects the payoff of both
players.

The channel has finite sets (or ``alphabets'') $X$ and $Y$
of input and output symbols, with noise given by transition
probabilities $p(y|x)$ for each $x$ in $X$ and $y$ in~$Y$.
The channel is used $n$ times independently without
feedback.  
When an input $x=(x_1,\ldots,x_n)$ is transmitted through
the channel, it is altered to an output $y=(y_1,\ldots,y_n)$
according to the probability $p(y|x)$ given by
\begin{equation}
\label{pyx}
p(y|x)=\prod_{j=1}^n p(y_j|x_j).
\end{equation}
This is the standard model of a memoryless noisy channel as
considered in information theory
(see Cover and Thomas, 1991; Gallager, 1968; MacKay, 2003).

The sender's strategy is to encode state $i$ by means of a
coding function or {\em code\/} $c:\Omega\to X^n$, which we write as
$c(i)=x^i$.
We call $x^i$ the {\em codeword\/} or {\em message\/} for
state~$i$ in $\Omega$, which the sender transmits as input
to the channel.
The code $c$ is completely specified by the
list of $M$ codewords $x^0,x^1,\ldots,x^{M-1}$, which is
called the {\em codebook\/}. 

The receiver's strategy is to decode the channel output~$y$,
given by a probabilistic {\em decoding function\/} 
\begin{equation}
\label{dec}
d:Y^n\times \Omega\to\reals,
\end{equation}
where $d(y,i)$ is the probability that $y$ is decoded
as~$i$.

If the receiver decodes the channel output as the
state~$i$ chosen by nature, then sender and receiver get
positive payoff $U_i$ and $V_i$, respectively, otherwise
both get payoff zero.
The incentives of sender and receiver are fully aligned in
the sense that they always prefer that the state is
communicated successfully.
However, the importance of that success may be different for
sender and receiver depending on the state.
The channel transition probabilities, the transmission
length~$n$, and the prior probabilities $q_i$ and utilities
$U_i$ and $V_i$ for $i$ in $\Omega$ are commonly known to
the players. 

\begin{definition}
\label{d-nash}
Consider an encoding function $c:\Omega\to X^n$ and a 
probabilistic decoding function $d$ in $(\ref{dec})$.
If the pair $(c,d)$ defines a Nash equilibrium,
then $c$ is called a {\em Nash code}.
The expected payoffs to sender and receiver are denoted by
$U(c,d)$ and $V(c,d)$, respectively.
\end{definition}

In order to obtain a Nash equilibrium $(c,d)$, receiver and
sender have to play mutually best responses.
The equilibrium property, and whether $c$ is called a Nash
code as part of such an equilibrium, may depend on the
particular best response $d$ of the receiver.

A code $c$ defines the sender's strategy.
A best response of the receiver is the following.
Given that he receives channel output $y$ in $Y^n$, the
probability that codeword $x^i$ has been sent is, by Bayes's
law, $q_i\,p(y|x^i)/\prob(y)$, where $\prob(y)$ is the
overall probability that $y$ has been received.
The factor $1/\prob(y)$ can be disregarded in the
maximization of the receiver's expected payoff.
Hence, a best response of the receiver is to choose with
positive probability $d(y,i)$ only states~$i$ so that 
$q_iV_i\,p(y|x^i)$ is maximal, that is, so that $y$ belongs
to the set $Y_i$ defined by 
\begin{equation}
\label{Yi}
Y_i=\{y\in Y^n\mid q_i V_i\,p(y|x^i)\ge q_k V_k\, p(y|x^k)
~~\forall k\in\Omega\}.
\end{equation}
Hence, the best response condition for the receiver states
that for any  $y\in Y^n$ and $i\in \Omega$
\begin{equation}
\label{arbtie}
d(y,i)>0 \quad\Rightarrow\quad y \in Y_i\,.
\end{equation}
If $V_i=1$ for all $i\in \Omega$, then this decoding rule is
known as MAP or {\em maximum a posteriori decoding} (MacKay,
2003, p.~305).
If the receiver has different positive utilities $V_i$ for
different states~$i$, then the receiver's best response
maximizes $q_iV_i\,p(y|x^i)$.
We call the product $q_iV_i$ the {\em weight} for
state~$i$.
One could assume $V_i=1$ for all~$i$ and only vary $q_i$ in
place of the weight,
but then it seems artificial to allow separate utilities
$U_i$ for the sender, because we want to study the
Nash property with respect to the optimality of codes for
receiver and sender.
For that reason we keep three parameters $q_i$, $U_i$ and
$V_i$ for each state~$i$.

We say that for a given channel output $y$, there is a {\em
tie\/} between two states $i$ and $k$ (or the states are
{\em tied\,}) if $y\in Y_i\cap Y_k$.
If there are never any ties, then the sets $Y_i$ for
$i\in\Omega$ are pairwise disjoint, and the best-response
decoding function is deterministic and unique according
to~(\ref{arbtie}).
If there are ties, then a natural way to break them is to
choose any of the tied states with equal probability.
For that reason we consider probabilistic decoding
functions.
On the sender's side, we only consider deterministic
encoding strategies.

We sometimes refer to the sets $Y_i$ for $i\in \Omega$ as a
``partition'' of $Y^n$, which constrains the receiver's
best-response decoding as in (\ref{arbtie}), even though
some of these sets may be empty, and they may not always be
disjoint if there are ties. 
In any case, $Y^n=\bigcup_{i\in \Omega}Y_i$.

Suppose that the receiver decodes the channel output with
$d$ according to (\ref{Yi}) and (\ref{arbtie}) for the given
code $c$ with $c(i)=x^i$.
Then $(c,d)$ is a Nash equilibrium if and only if, for any
state~$i$, it is optimal for the sender to transmit $x^i$
and not any other $\hat x$ in $X^n$ as a message.
When sending $\hat x$, the expected payoff to the sender in
state~$i$ is
\begin{equation}
\label{pay}
U_i\sum_{y\in Y^n}p(y|\hat x)\,d(y,i).
\end{equation}
When maximizing (\ref{pay}) as a function of $\hat x$,
the utility $U_i$ to the sender
does not matter as long as it is positive; given that the
state is~$i$, the sender only cares about the expected
probability that the channel output~$y$ is decoded as~$i$.
We summarize these observations as follows.

\begin{proposition}
\label{p-send}
The code $c$ with decoding function $d$ is a Nash
code if and only if the receiver decodes channel outputs
according to $(\ref{Yi})$ and $(\ref{arbtie})$, and if and
only if in every state~$i$ the sender transmits codeword
$c(i)=x^i$ which fulfills for any other possible channel
input $\hat x$ in $X^n$ 
\begin{equation}
\label{sender}
\sum_{y\in Y^n} p(y|x^i)\,d(y,i)
\ge
\sum_{y\in Y^n} p(y|\hat x)\,d(y,i)\,. 
\end{equation}
\end{proposition}

\section{Examples of codes that are not Nash}
\label{s-examples}

This section presents introductory examples of channels that
are used once ($n=1$) and that illustrate that the Nash
equilibrium condition does not hold automatically.
At the end of this section, we show in
Proposition~\ref{p-diag} that, under certain assumptions,
the Nash property holds when all channel symbols are used
for transmission.

For our first example, consider a channel with three
symbols, $X=Y=\{0,1,2\}$, which is used only once ($n=1$), 
with the following transition probabilities: 
\begin{equation}
\renewcommand{\arraystretch}{1.5}
\label{ternary}
\begin{tabular}{|r|lll|}
\hline
& &  \hfil $y$ & \\[-1ex]
\raise 2ex\hbox{$p(y|x)$} & \hfil 0 & \hfil 1 & \hfil 2 \\
\hline
0 & 0.7 & 0.15 & 0.15\\
~$x$ \hfill 1 & 0.25 & 0.5 & 0.25 \\
2 & 0.2 & 0.2 & 0.6 \\
\hline
\end{tabular}
\end{equation} 
Suppose that there are two states ($m=2$) and that nature
chooses the two states from $\Omega=\{0,1\}$ with uniform
priors $q_0=q_1=1/2$.
The sender's utilities are $U_0=2$ when the state is $0$
and $U_1=8$ when the state is 1, and the receiver's
utilities are $V_0=6$, $V_1=4$.

Consider the codebook $c$ with $c(0)=x^0=0$ and $c(1)=x^1=1$,
so the sender codifies the two states of nature as the two
symbols $0$ and $1$, respectively. 
Given the parameters of this game and the sender's
strategy~$c$, the receiver's strategy assigns to each output
symbol in $\{0,1,2\}$ one state.
The following table (\ref{nonash}) gives the expected payoff
$q_i V_i \,p(y|x^i)$ for the receiver when the state is $i$
and the output symbol is $y$.
\begin{equation}
\renewcommand{\arraystretch}{1.5}
\label{nonash}
\begin{tabular}{|r|lll|}
\hline
& &  \hfil $y$ & \\[-1ex]
\raise 2ex\hbox{$q_i V_i \,p(y|x^i)$} & \hfil 0 &  \hfil 1 & \hfil 2 \\
\hline
0 & 2.1 & 0.45 & 0.45\\[-1ex]
\raise 2ex\hbox{~~$i$} \hfill
1 & 0.5 & 1 & 0.5 \\
\hline
\end{tabular}
\end{equation}
This shows how to find the receiver's best response and the
sets $Y_i$ in (\ref{Yi}).  
For each channel output $y$, the receiver chooses the state
$i$ with highest expected payoff.
Hence, he decodes the channel output $0$ as state $0$ because
$q_0V_0\,p(0|x^0)=2.1>0.5=q_1V_1\,p(0|x^1)$.
In the same way, he decodes both channel outputs $1$ and $2$
as state~$1$.
Here there are no ties, so the two sets $Y_0$ and
$Y_1$ are disjoint, and the receiver's best response is
unique and deterministic.
That is, the receiver's best response $d$ is given by
$d(y,i)=1$ if $y \in Y_i$, where $Y_0=\{0\}$ and
$Y_1=\{1,2\}$, and by $d(y,i)=0$ otherwise.

A shorter form of obtaining (\ref{nonash}) from
the channel transition probabilities in (\ref{ternary}) is
shown in (\ref{weighted}), which is (\ref{ternary}) with
each row prefixed by the weight $q_iV_i$ when the channel
input for that row is used as codeword $x^i$.  
Multiplying the channel probabilities with these weights
gives (\ref{nonash}), and a box surrounds $p(y|x^i)$ if
output $y$ is decoded as state~$i$. 
These boxes therefore also show the sets $Y_i$ if there are
no ties, as in the present case; in the case of ties, and
deterministic decoding, they show the state that is 
actually decoded by the receiver.  
\begin{equation}
\renewcommand{\arraystretch}{1.5}
\label{weighted}
\begin{tabular}{c|r|lll|}
\cline{2-5}
&& &  \hfil $y$ & \\[-1ex]
\raise 2ex\hbox{$q_iV_i$} &
\raise 2ex\hbox{$p(y|x)$} &
\hfil 0 & \hfil 1 & \hfil 2 \\
\cline{2-5}
3&0 & \fbox{0.7} & ~0.15 & ~0.15\\
2&~$x$ \hfill 1 & ~0.25 & \fbox{0.5} & \fbox{0.25} \\
&2 & ~0.2 & ~0.2 & ~0.6 \\
\cline{2-5}
\end{tabular}
\end{equation} 

With the help of Proposition~\ref{p-send}, it is easy to see
from (\ref{weighted}) that this code $c$ is not a Nash code.  
For $i=0$ and $x^0=0$, we have
$\sum_{y \in Y}p(y |0)d(y,0)=0.7$,
which is the maximum of the column entries $p(y|x)$ for
$y=0$ in (\ref{weighted}), so here the sender cannot improve
her payoff by transmitting any $\hat x$ instead of~$x^i$.
However, for $i=1$ we have
$\sum_{y \in Y}p(y |1)d(y,1)=0.5+0.25=0.75<0.8
=0.2+0.6=\sum_{y \in Y}p(y |2)d(y,2)$, so (\ref{sender})
does not hold when $x^i=1$ and $\hat x=2$ and the sender can
improve her payoff by sending $\hat x$ instead of~$x^i$.

Is there a Nash code for the channel in (\ref{ternary}) when
$\Omega=\{0,1\}$ and for the described priors and utilities?
First, a simple and trivial Nash code is to map both states
to the same, arbitrary channel input, $x^0=x^1$.
Then every channel output results from the same row (for
that input) in (\ref{ternary}) and, because $q_0V_0>q_1V_1$,
will be decoded as state~0.
The sender cannot improve her payoff because the receiver in
effect ignores the uninformative channel output.
This is also called a ``babbling'' or ``pooling''
equilibrium, which is a Nash equilibrium for any channel.

When the codewords are distinct ($x^0\ne x^1$), there are
six possible ways to choose them from the three channel
inputs.
Table~\ref{6codes} lists these codebooks $x^0,x^1$, shown in
the first column.
For each code, the receiver's best response is unique.
The best-response partition $Y_0,Y_1$ is shown in the second
column.
Using this partition, the third column gives the
probabilities $p(y \in Y_i\mid x^i)=\sum_{y \in Y_i}p(y|x^i)$
that the codeword $x^i$ is decoded correctly.
The overall expected payoffs to sender and receiver are
shown as $U$ and~$V$, with a box indicating the respective
maximum.  

\begin{table}[hbtp]
\caption{Possible codebooks $x^0,x^1$ with $x^0\ne x^1$ for
the channel (\ref{ternary}) and expected payoffs $U$ and $V$
to sender and receiver.}
\label{6codes}
\[
\renewcommand{\arraystretch}{1.5}
\begin{array}{|c|cc|cc|c|c|}
\hline
~~x^0,x^1~~ & ~~Y_0~~ & ~~Y_1~~ &
  p(y\in Y_0\mid x^0) & p(y\in Y_1\mid x^1) & U & V \\
\hline
0,1 & \{0\} & \{1,2\} & 0.70 & 0.75 & ~\fbox{3.70}~ & ~~3.60~~ \\
0,2 & \{0,1\} & \{2\} & 0.85 & 0.60 & ~~3.25~~ & 3.75 \\
1,0 & \{1,2\} & \{0\} & 0.75 & 0.70 & 3.55 & 3.65 \\
1,2 & \{0,1\} & \{2\} & 0.75 & 0.60 & 3.15 & 3.45 \\
2,0 & \{1,2\} & \{0\} & 0.80 & 0.70 & 3.60 & ~\fbox{3.80}~ \\
2,1 & \{0,2\} & \{1\} & 0.80 & 0.50 & 2.80 & 3.40 \\
\hline
\end{array}
\]
\end{table}


Similar to using (\ref{weighted}) for the codebook $0,1$, it
can be verified that the codebook $2,1$ is not a Nash code.
In addition, Table~\ref{6codes} shows directly that the
codebook $1,0$ is not a Nash code:
It has the same best response of the receiver (given by
$Y_0=\{1,2\}$ and $Y_1=\{0\}$) as the codebook $2,0$, but
a lower payoff to the sender (3.55 instead of 3.6), who can
therefore improve her payoff by changing $x^0=1$ to~$\hat
x=2$ (note that the receiver's reaction stays fixed).
Similarly, codebook $1,2$ has the same best response as
$0,2$, but a lower payoff to the sender (3.15 instead of
3.25).

Only the codebooks $2,0$ and $0,2$ in Table~\ref{6codes} are
Nash codes.
Apart from a direct verification, this follows from Theorems 
\ref{t-global} and \ref{t-opt}, respectively, which we will
discuss in the next section.

The fact that a code is not Nash seems to be due to the fact
that not all symbols of the channels are used for
transmission.
With some qualifications, this is indeed the case, as we
discuss in the remainder of this section.

Consider the channel in (\ref{ternary}) and suppose that
there are three states, $\Omega=\{0,1,2\}$.
However, even when each state is assigned to a different
input symbol, one can replicate the counterexample in
(\ref{weighted}) when the additional state $2$ has a weight
$q_2V_2$ that is too low.
For example, if priors are uniform as before ($q_i=1/3$) and
$V_0=6$, $V_1=4$, and $V_2=1$, then the channel outputs
would be decoded as before when state~2 is absent, with the
same lack of the Nash property.

Hence, one should require that all {\em output\/} symbols
are decoded differently.
However, the following example shows that this may still 
fail to give a Nash code:
\begin{equation}
\renewcommand{\arraystretch}{1.5}
\label{3diff}
\begin{tabular}{c|r|lll|}
\cline{2-5}
&& &  \hfil $y$ & \\[-1ex]
\raise 2ex\hbox{$q_iV_i$} &
\raise 2ex\hbox{$p(y|x)$} &
\hfil 0 & \hfil 1 & \hfil 2 \\
\cline{2-5}
0.35 &0 & \fbox{0.4} & ~0.3& ~0.3\\
0.35&~$x$ \hfill 1 & ~0.3  & \fbox{0.4} & ~0.3 \\
0.3\phantom5&2& ~0.05 & ~0.45 & \fbox{0.5} \\
\cline{2-5}
\end{tabular}
\end{equation} 
Suppose states 0, 1, 2 are encoded as 0, 1, 2 and have the
indicated weights 0.35, 0.35, 0.3, respectively.
Here, the row for channel input 1 has slightly higher weight
than for input~2, so because $0.35\times 0.4> 0.3\times
0.45$ the decoding function is just the identity.
However, for state~1 the sender can improve the probability
of correct decoding by deviating from $x^1=1$ to $\hat x=2$
because $0.4<0.45$.

In (\ref{3diff}), for any input $x$ the corresponding output
$y=x$ has the highest probability of arriving, but this is
not relevant for decoding.
With uniform priors and utilities, a reasonable condition
for the channel is ``identity decoding'', that is,
for any received output $y$, the maximum likelihood input is
$x=y$.
That is, suppose that
\begin{equation}
\label{diag}
X=Y,\qquad
p(y|y)>p(y|x)
\quad\hbox{for all }y\in Y,~x\in X,~x\ne y
\end{equation}
which says that each output symbol $y$ is more likely to have
been received correctly than in error.
This property is violated in (\ref{3diff}), but if it holds
then the following proposition applies. 

\begin{proposition}
\label{p-diag}
Consider a channel with input and output alphabets $X$ and
$Y$ so that $(\ref{diag})$ holds.
Let $c$ be a code so that each channel output is decoded as
coming from a different channel input $x^i$ with a
deterministic best-response decoding function~$d$.
Then $(c,d)$ is a Nash equilibrium and $c$ is a Nash code.  
Every output $y$ is decoded as a state $i$ so that $x^i=y$
and so that $q_iV_i$ is maximal.
\end{proposition}

\proof
By assumption, $X=Y$ and the map $\phi:Y\to X$ defined by
$\phi(y)=x^i$ if $d(y,i)=1$ is injective and hence a
bijection.
Suppose $\phi$ is not the identity map, so it has a cycle
of length $l>1$, which by permuting $\Omega$ we assume
as coming form the first $l$ states
$x^0,x^1,\ldots,x^{l-1}$, that is, $\phi(x^j)=x^{j+1\mod l}$
for $0\le j<l$.  
So channel output $x^0$ is decoded as state~1 because
$\phi(x^0)=x^1$, output $x^1$ is decoded as state~2, and so
on.
Because $d$ is a best-response decoding function,
$q_0V_0\,p(x^0|x^0)\le q_1V_1\,p(x^0|x^1)$ and therefore 
\[
q_0V_0\le q_1V_1\frac{p(x^0|x^1)}{p(x^0|x^0)}<q_1V_1
\]
by (\ref{diag}).
In the same manner, $q_1V_1<q_2V_2<\cdots<
q_{l-1}V_{l-1}<q_0V_0$, a contradiction.

So $\phi$ is identity map.
Consider any state~$i$.
If $d(y,i)=0$ for all outputs $y$, then (\ref{sender}) holds
trivially.
Otherwise, channel output $x^i$ is decoded as state~$i$ and
(\ref{sender}) holds because 
\[
\sum_{y\in Y} p(y|x^i)\,d(y,i)=p(x^i|x^i)\ge p(x^i|\hat x)
= \sum_{y\in Y} p(y|\hat x)\,d(y,i)
\]
by (\ref{diag}).
So $c$ is a Nash code.
The encoding function $c$ is surjective because every input
$x^i$ occurs as a possible decoding as a state~$i$. 
However, if $|\Omega|>|X|$, then $c$ is not injective.
If $x^i=x^k$, then $d(y,i)=1$ requires that $x^i=y$ and that
$q_iV_i\ge q_kV_k$ by the best-response condition (in fact
for any state~$k$), as claimed.
\endproof

In many contexts, in particular when a channel is used
repeatedly, a code does not use all possible channel inputs
in order to allow for redundancy and error correction.
Sufficient conditions for Nash codes beyond
Proposition~\ref{p-diag} are therefore of interest.

\section{Receiver-optimal codes}
\label{s-recopt}

In this section, we show that every code that maximizes the
receiver's payoff is a Nash code.
The proof implies that this holds also if the receiver's
payoff is locally maximal, that is, when changing only a
single codeword, and the corresponding best response of the
receiver, at a time.
Finally, we discuss the connection with potential functions.

In the example (\ref{weighted}), changing the codebook $c$
to $c'$ where $c'(1)=\hat x=2$ improves the sender payoff
from $U(c,d)$ to $U(c',d)$, where $d$ is the receiver's
best-response decoding for code~$c$.
In addition, it is easily seen that the receiver payoff also
improves from $V(c,d)$ to $V(c',d)$, and his payoff
$V(c',d')$ for the best response $d'$ to $c'$ is possibly
even higher.
This observation leads us to a sufficient condition for Nash
codes.

\begin{definition}
\label{d-opt}
A {\em receiver-optimal code} is a code $c$ with highest
expected payoff to the receiver, that is, so that
\[
V(c,d)\ge V(\hat c, \hat d)
\]
for any other code $\hat c$,
where $d$ is a best response to $c$ 
and $\hat d$ is a best response to $\hat c$.
\end{definition}

Note that in this definition, the expected payoff $V(c,d)$
(and similarly $V(\hat c,\hat d)$) does not depend on the
particular best-reponse decoding function $d$ in case $d$
is not unique when there are ties, because the receiver's
payoff is the same for all best responses~$d$.

The following is the central theorem of this section.
It is proved in three simple steps,\footnote{ We are indebted
to Drew Fudenberg who suggested steps two and three.} which
give rise to a generalization that we discuss afterwards,
along with examples and further observations.

\begin{theorem}
\label{t-global}
Every receiver-optimal code is a Nash code.
\end{theorem}

\begin{proof}
Let $c$ be a receiver-optimal code with codebook
$x^0,x^1,\ldots,x^{M-1}$,
and let $d$ be an arbitrary decoding function.
Suppose there exists a code $\hat{c}$ with codebook
$\hat{x}^0,\hat{x}^1,\ldots,\hat{x}^{M-1}$ so that
$U(\hat c,d)>U(c,d)$, that is, 
\begin{equation}
\label{hatc}
\sum_{i \in \mathstrut\Omega} q_i U_i
\sum_{y \in \mathstrut Y^n} p(y | \hat{x}^i) d(y,i) > 
\sum_{i \in \mathstrut\Omega} q_i U_i
\sum_{y \in \mathstrut Y^n} p(y | x^i) d(y,i).
\end{equation}
If $d$ is a best response to $c$ according to $(\ref{Yi})$
and $(\ref{arbtie})$, then (\ref{hatc}) holds for some
$\hat c$ if and only if $c$ is not a Nash code, so suppose
that $c$ is not a Nash code;
however, the following steps one and two apply for any~$d$.

Step one: Clearly, (\ref{hatc}) implies%
\footnote{ This claim follows also directly from
Proposition~\ref{p-send}, but we want to refer later to
(\ref{hatc}) as well.} 
that there exists at least one $i \in \Omega$ so that 
\begin{equation}
\label{oneimprov}
\sum_{y \in \mathstrut Y^n} p(y | \hat{x}^i) d(y,i) >
\sum_{y \in \mathstrut Y^n} p(y  | x^i) d(y,i).
\end{equation}
Consider the new code $c'$ which coincides with $c$ except
for the codeword for state~$i$, where we set
$c'(i)=\hat{x}^i$.
So the codebook for $c'$ is
$x^0,\ldots,x^{i-1},\hat{x}^i,x^{i+1},\ldots,x^{M-1}$.
By (\ref{oneimprov}), we also have
\begin{equation}
\label{deviation}
\begin{array}{rl}
U(c',d)=&\displaystyle
\sum_{j\in \mathstrut\Omega,~j\neq i} q_j U_j
\sum_{y \in \mathstrut Y^n} p(y | x^j) d(y,j)+
q_i U_i \sum_{y \in \mathstrut Y^n} p(y| \hat{x}^i) d(y,i)\\[4ex]
>&\displaystyle
\sum_{j \in \mathstrut\Omega} q_j U_j
\sum_{y \in \mathstrut Y^n} p(y| x^j) d(y,j)
= U(c,d).\hfill
\end{array}
\end{equation}
Step two:
In the same manner, (\ref{oneimprov}) implies an improvement
of the receiver function, that is,
\begin{equation}
\label{improve}
V(c',d)>V(c,d).
\end{equation} 
Step three:
Let $d$ be the best response to $c$ and
let $d'$ be the best response to $c'$.
With (\ref{improve}), this implies
\[
V(c',d')\ge V(c',d)>V(c,d).
\]
Hence, code $c'$ has higher expected receiver payoff than $c$.
This contradicts the assumption that $c$ is a receiver-optimal code.
\end{proof}

In Table~\ref{6codes}, the codebook $2,0$ is
receiver-optimal, and a Nash code in agreement with
Theorem~\ref{t-global}.

We have shown that the codebook $0,1$ in Table~\ref{6codes}
is not a Nash code.
Note, however, that this is the code with highest sender
payoff.
Hence, a ``sender-optimal'' code is not necessarily a Nash
code.  
The reason is that, because sender and receiver have
different payoffs for the two states, the sender prefers the
code with large partition class $Y_1$ for state~1, but then
can deviate to a better, unused message within $Y_1$.
(Note that the sender's payoff only improves when the
receiver's response stays fixed; with best-response
decoding, the code $0,2$ has a worse payoff $U$ to the
sender than $0,1$.)

In Table~\ref{6codes},
the code $c$ with codebook $0,2$ is also seen to be a Nash
code with the help of Table~\ref{6codes} according to the
proof of Theorem~\ref{t-global}.
Namely, it suffices to look for profitable sender deviations
$c'$ where only one codeword is altered, which would also
imply an improvement to the receiver's payoff from $V(c,d)$
to $V(c',d)$, and hence certainly an improvement to his
payoff $V(c',d')$ where $d'$ is the best response to~$c'$.
For the two possible codes $c'$ given by $1,2$ and $0,1$,
the receiver payoff $V$ does not improve according
to~Table~\ref{6codes}, so $c$ is a Nash code.
By this reasoning, any ``locally'' receiver-optimal
code, according the following definition, is also a Nash
code, as stated afterwards in Theorem~\ref{t-opt}.

\begin{definition}
\label{d-local}
A {\em locally receiver-optimal code} is a code $c$ so that
no code $c'$ that differs from $c$ in only a single codeword
gives higher expected payoff to the receiver.
That is, for all $c'$ with $c'(i)\ne c(i)$ for some
state~$i$, and $c'(k)=c(k)$ for all $k\ne i$, 
\[
V(c,d)\ge V(c', d')
\]
where $d$ is a best response to~$c$ 
and $d'$ is a best response to~$c'$.
\end{definition}

\begin{theorem}
\label{t-opt}
Every locally receiver-optimal code is a Nash code. 
\end{theorem}

\begin{proof}
Apply the proof of Theorem \ref{t-global} from Step~two
onwards.
\end{proof}

Clearly, every receiver-optimal code is also locally
receiver-optimal, so Theorem~\ref{t-global} can be
considered as a corollary to the stronger Theorem~\ref{t-opt}.

Local receiver-optimality is more easily verified than
global receiver-optimality, because much fewer codes $c'$
have to be considered as possible improvements for the
receiver payoff according to Definition~\ref{d-local}.
A locally receiver-optimal code can be reached by iterating
profitable changes of single codewords at a time.
This simplifies the search for a (nontrivial) Nash code.

To conclude this section, we consider the connection to
{\em potential games} which also allow for iterative
improvements in order to find a Nash equilibrium.
As in Monderer and Shapley (1996, p.~127),
consider a game in strategic form with finite player
set~$N$, and pure strategy set $S_i$ and utility function 
$u^i$ for each player~$i$.
Then the game has an (ordinal) {\em potential function}
$P: \prod_{j\in N}S_j\to\reals$ if for all $i\in N$ and
$s^{-i}\in \prod_{j\ne i}S_j$ and $s^i,\hat s^{\,i}\in S_i$,
\begin{equation}
\label{potential}
u^i(s^{-i},\hat s^{\,i})> u^i(s^{-i},s^i)
\quad
\Leftrightarrow
\quad
P(s^{-i},\hat s^{\,i})> P(s^{-i},s^i).
\end{equation}
The question is if in our game, the receiver's payoff is a
potential function.\footnote{ We thank Rann Smorodinsky for
raising this question.}
The following proposition gives an answer.

\begin{proposition}
\label{p-pot}
Consider the game with $M+1$ players where for each state
$i$ in $\Omega$, a separate agent~$i$ transmits a codeword
$c(i)$ over the channel, which defines a function
$c:\Omega\to X^n$, and where the receiver decodes each
channel output with a decoding function $d$ as before.
Each agent receives the same payoff $U(c,d)$ as the original
sender. 
Then
\parskip0pt
\begin{itemize}
\item[(a)]
Any Nash equilibrium $(c,d)$ of the $(M+1)$-player game is a 
Nash equilibrium of the original two-player game, and vice
versa.
\item[(b)]
The receiver's expected payoff is a potential function for
the $(M+1)$-player game.
\item[(c)]
The receiver's expected payoff is not necessarily a
potential function for the original two-player game.
\end{itemize}
\end{proposition}

\begin{proof}
Every profile $c$ of $M$ strategies for the agents in the
$(M+1)$-player game can be seen as a sender strategy in
the original game, and vice versa.
To see (a), let $(c,d)$ be a Nash equilibrium of the
$(M+1)$-player game.
If there was a profitable deviation $\hat c$ from $c$ for
the sender in the two-player game as in (\ref{hatc}), then
there would also be a profitable deviation $c'$ that changes
only one codeword $c(i)$ as in (\ref{deviation}), which is a
profitable deviation for agent~$i$, a contradiction.
The ``vice versa'' part of (a) holds because any profitable
deviation of a single agent is also a deviation for the
sender in the original game.

Assertion (b) holds because for any~$i$ in $\Omega$,
(\ref{deviation}) is, via (\ref{oneimprov}), equivalent to
(\ref{improve}).

To see (c), consider the example (\ref{ternary}) with
$c$ and $\hat c$ given by the codebooks $1,0$ and $2,1$, 
respectively, and $d$ decoding channel outputs $y=0,1,2$
as states $0,0,1$, respectively.
Then the payoffs to sender and receiver are
\[
\arraycolsep2pt
\renewcommand{\arraystretch}{1.2}
\begin{array}{rlll}
U(c,d) &= q_0U_0(p(0|1)+p(1|1))+q_1U_1\,p(2|0)
&=1\times (0.25+0.5)+4\times0.15 &= 1.35\\
V(c,d) &= q_0V_0\,(p(0|1)+p(1|1))+q_1V_1\,\,p(2|0)
&=3\times (0.25+0.5)+2\times0.15 &= 2.55\\
U(\hat c,d) &= q_0U_0(p(0|2)+p(1|2))+q_1U_1\,p(2|1)
&=1\times (0.2+0.2)+4\times0.25 &= 1.4\\
V(\hat c,d) &= q_0V_0\,(p(0|2)+p(1|2))+q_1V_1\,\,p(2|1)
&=3\times (0.2+0.2)+2\times0.25 &= 1.7\\
\end{array}
\]
which shows that (\ref{potential}) does not hold with $u^i$
as sender payoff and $P$ as receiver payoff, because these
payoffs move in opposite directions when changing the
sender's strategy from $c$ to $\hat c$, for this~$d$.  
\end{proof}

A global maximum of the potential function gives a Nash
equilibrium of the potential game (Monderer and Shapley,
1996, Lemma~2.1).
Hence, (a) and (b) of Proposition~\ref{p-pot} imply that
a maximum of the receiver payoff defines a Nash equilibrium,
as stated in Theorem~\ref{t-global}.
It is also known that a ``local'' maximum of the potential
function defines a Nash equilibrium (Monderer and Shapley,
1996, footnote~4). 
However, this does not imply Theorem~\ref{t-opt}.
The reason is that in a local maximum of the potential
function, the function cannot be improved by unilaterally
changing a single player's strategy.
In contrast, in a locally receiver-optimal code, the
receiver's payoff cannot be improved by changing a single
codeword {\em together} with the receiver's best
response.
As a trivial example, any ``babbling'' Nash code for
(\ref{ternary}) where $x^0=x^1$ is not locally
receiver-optimal, but is a ``local maximum'' of the receiver
payoff.

In a potential game, improvements of the potential function
can be used for dynamics that lead to Nash equilibria.
For our games, the study of such dynamics may be an
interesting topic for future research.  

\section{Binary channels and monotonic decoding}
\label{s-bin}

Our next main result (stated in the next section) concerns
the important {\em binary channel} with $X=Y=\{0,1\}$.
The two possible symbols 0 and 1 for a single use of the
channel are called {\em bits}.
The binary channel is the basic model for the transmission
of digital data and of central theoretical and practical
importance 
in information theory (see, for example, Cover
and Thomas, 1991, or MacKay, 2003).

We assume that the channel errors $\eps_0=p(1|0)$ and
$\eps_1=p(0|1)$ fulfill  
\begin{equation}
\label{errors}
\eps_0>0,
\qquad
\eps_1>0,
\qquad
\eps_0+\eps_1<1,
\end{equation}
where $\eps_0+\eps_1<1$ is equivalent to either of the
inequalities, equivalent to (\ref{diag}),
\begin{equation}
\label{error2}
1-\eps_0>\eps_1,
\qquad
1-\eps_1>\eps_0.
\end{equation}
These assert that a
received bit 0 is more likely to have been sent as 0 (with
probability $1-\eps_0$) than sent as bit 1 and received with
error (with probability~$\eps_1$), and
similarly
that a received bit 1 is more likely to have been sent as 1
than received erroneously.
It may still happen that bit 0, for example, is transmitted
with higher probability incorrectly than correctly, for
example if $\eps_0=3/4$ and $\eps_1=1/8$.

Condition (\ref{errors}) can be assumed with very little
loss of generality.
If $\eps_0=\eps_1=0$ then the channel is error-free and
every message can be decoded perfectly.
If $\eps_0+\eps_1=1$ then the channel output is independent
of the input and no information can be transmitted.
For $\eps_0+\eps_1>1$ the signal is more likely to be
inverted than not, so that one obtains (\ref{errors}) by
exchanging 0 and~1 in $Y$.

Condition (\ref{errors}) does exclude the case
of a ``Z-channel'' that has only one-sided errors, that is, 
$\eps_0=0$ or $\eps_1=0$. 
We assume instead that this is modelled by vanishingly small
error probabilities, in order to avoid channel outputs $y$
in $Y^n$ that cannot occur for some inputs~$x$ when
$\eps_0=0$ or $\eps_1=0$. 
With (\ref{errors}), every channel output $y$ has positive,
although possibly very small, probability.  

The binary channel is {\em symmetric\/} when
$\eps_0=\eps_1=\eps>0$, where $\eps<1/2$ by (\ref{errors}).

The binary channel is used $n$ times independently.
A code $c:\Omega\to X^n$ for $X=\{0,1\}$ is also
called a {\em binary code}.
Our main result about binary codes (Theorem~\ref{t-bin} below)
implies that {\em any\/} binary code is a Nash code,\footnote{ %
Hern\'andez, Urbano, and Vila (2010) show that for a binary
noisy channel, the decoding rule of ``joint typicality''
used in a standard proof of Shannon's channel coding theorem
(Cover and Thomas, 1991, Section~8.7) may not define a Nash
equilibrium.
}
 provided the decoding is {\em monotone}.
This monotonicity condition concerns how the receiver
resolves ties when a received channel output $y$ can be
decoded in more than one way.

We first consider an example of a binary code that shows
that the equilibrium property may depend on how the receiver
deals with ties.  
Assume that the channel is symmetric with error
probability~$\eps$.
Let $M=4$, $n=3$, and consider the codebook
$x^0,x^1,x^2,x^3$ given by $000,100,010,001$.
All four states $i$ have equal prior probabilities $q_i=1/4$
and equal sender and receiver utilities $U_i=V_i=1$.
The sets $Y_i$ in (\ref{Yi}) are given by
\begin{equation}
\label{Y0123}
\begin{array}{ll}
Y_0=\{000\},
&
Y_2=\{010,011,110,111\},
\\
Y_1=\{100,101,110,111\},
~~
&Y_3=\{001,011,101,111\}.
\\
\end{array}
\end{equation}
This shows that for any channel output $y$ other than an
original codeword $x^i$, there are ties between at least two
states.
For example, $110\in Y_1\cap Y_2$ because $110$ is received
with probability $\eps(1-\eps)^2$ for $x^1$ and $x^2$ as
channel input.  
For $y=111$, all three states $1,2,3$ are tied.

\begin{figure}[hbt]
\begin{center}
\includegraphics[width=6.5cm]{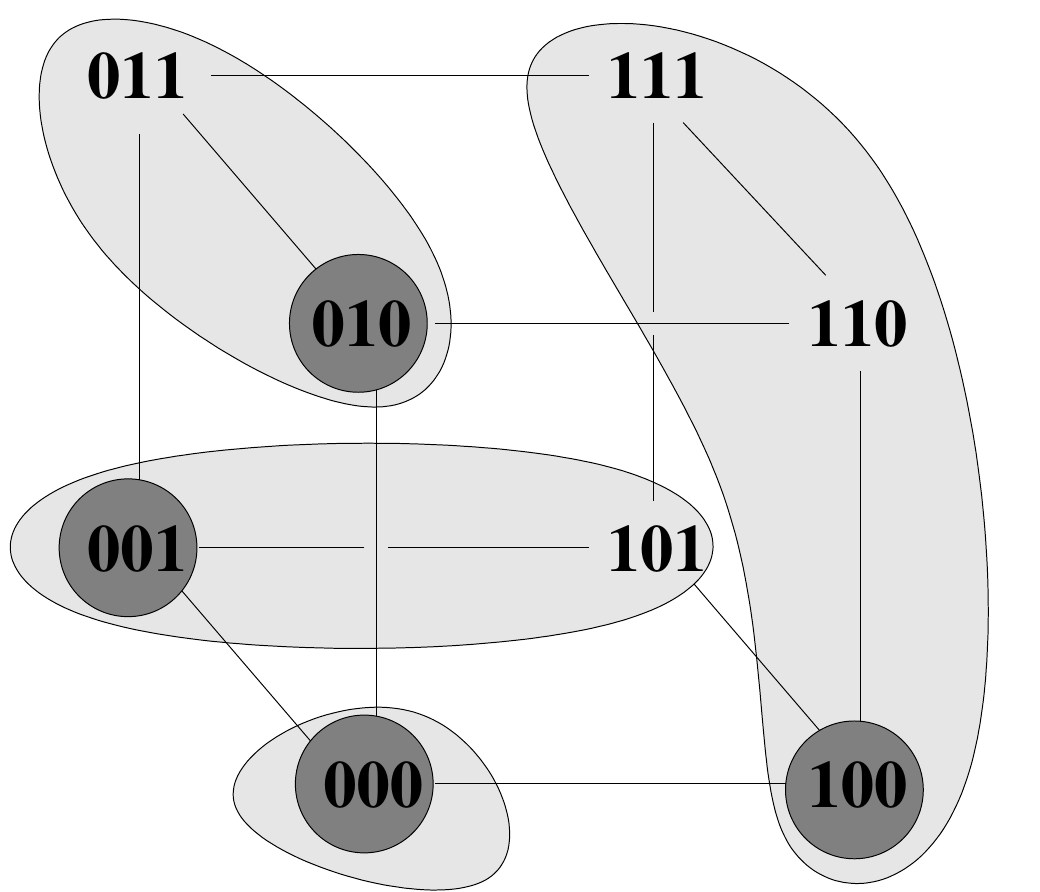}
\hfill
\includegraphics[width=6.5cm]{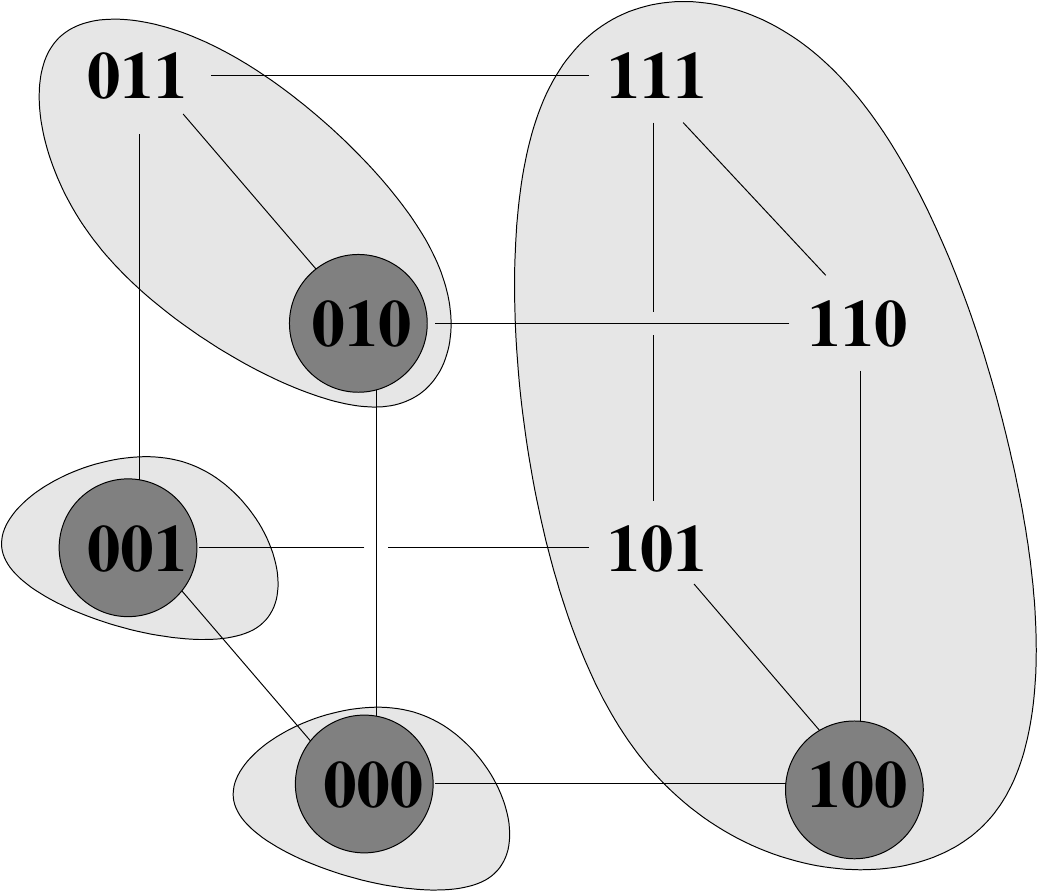}
\end{center}
\caption{Binary code with four codewords 000, 100, 010, 001,
with non-monotonic decoding (left) and monotonic decoding
(right, discussed in Section~\ref{s-mon}).
The light-grey sets indicate how a channel output is decoded.
}
\label{binary}
\end{figure}

Consider first the case that the receiver decodes the
channel outputs $110, 011, 101$ as states $1,2,3$,
respectively, that is, according to
\begin{equation}
\label{circular} 
d(110,1)=1,\quad d(011,2)=1,\quad d(101,3)=1.
\end{equation}
We claim that this cannot be a Nash code, irrespective of
the decoding probabilities $d(111,i)$ which can be positive
for any $i=1,2,3$ by (\ref{Y0123}).
The situation is symmetric for $i=1,2,3$, so assume that
$d(111,i)$ is positive when $i=1$;
the case of a deterministic decoding where $d(111,1)=1$ is
shown on the left in Figure~\ref{binary}.
Then the receiver decodes $y$ as state~1 with positive
probability when $y$ equals $100$, $110$, or $111$.
When $x^1=100$ is sent, these channel outputs 
are received with probabilities 
$(1-\eps)^3$, $\eps(1-\eps)^2$, and $\eps^2(1-\eps)$,
respectively, so the sender payoff is
\[
(1-\eps)^3+\eps(1-\eps)^2+\eps^2(1-\eps)\,d(111,1)
\]
in (\ref{pay}).
Given this decoding, the sender can improve her payoff in
state~1 by sending $\hat x=110$ rather than $x^1=100$ because
then the probabilities of the channel outputs $100$ and
$110$ are just exchanged, whereas the probability that
output $111$ is decoded as state~1 increases to
$\eps(1-\eps)^2\,d(111,1)$;
that is, given this decoding, sending $\hat x=110$ is more likely
to be decoded correctly as state~1 than sending $x^1=100$.
This violates (\ref{sender}).

The problem with the decoding in (\ref{circular}) is that
when the receiver is tied between states 1, 2, and~3 when
the channel output is $\hat y=111$, he decodes $\hat y$ as
state~1 with positive probability $d(111,1)$, but when he is
tied between even fewer states 1 and 3 when receiving
$y=101$, that decoding probability $d(101,1)$ decreases to
zero.
This violates the following {\em monotonicity\/} condition.

\begin{definition}
\label{d-mono}
Consider a codebook with codewords $x^i$ for $i\in \Omega$.
For a channel output $y$, let $T(y)$ be the set of tied states 
according to
\begin{equation}
\label{defT}
T(y)=\{l\in \Omega\mid y\in Y_l\}.
\end{equation}
Then a decoding function $d$ in $(\ref{dec})$ is called {\em
monotonic} if it is a best response decoding function with
$(\ref{Yi})$ and $(\ref{arbtie})$ and if for all
$y,\hat y\in Y^n$ and states~$i$, 
\begin{equation}
\label{mono}
i\in T(y)\subseteq T(\hat y)
~~\Rightarrow~~
d(y,i)\ge d(\hat y,i).
\end{equation}
Furthermore, $d$ is called {\em consistent} if 
\begin{equation}
\label{consistent}
i\in T(y)=T(\hat y)
~~\Rightarrow~~
d(y,i)= d(\hat y,i).
\end{equation}
\end{definition}

Condition (\ref{mono}) states that the probability of decoding the
channel output as state~$i$ can only decrease when the set
of tied states increases.
Condition (\ref{consistent}) states that the decoding
probability $d(y,i)$ of state~$i$ may only depend on the set
$T$ of states that are tied with~$i$, but not on the
received channel output~$y$.
Clearly, monotonicity implies consistency.
We will show that for certain channels, in particular the
binary channel, monotonic decoding gives a Nash code.
However, for consistent decoding this is not the case.
For example, the decoding shown in the left picture of
Figure~\ref{binary} is consistent because no two channel
outputs have the same set of tied states, but the Nash
property is violated.

Monotonic decoding functions exist, for example by breaking
ties uniformly at random according to $d(y,i)=1/|T(y)|$ for
$i\in T(y)$.  
We study the monotonicity condition in Definition~\ref{d-mono}
in more detail in later sections.

\section{Nash codes for input symmetric channels}
\label{s-insy}

In this section, we state and prove our main result,
Theorem~\ref{t-bin} below, about binary codes.
It turns out that it also applies to the following
generalization of discrete channels where the error
probability $\varepsilon_y$ of receiving an incorrect output
symbol $y$ only depends on~$y$ but not on the input.

\begin{definition}
\label{d-insym}
A discrete channel is {\em input symmetric} if $X=Y$
and there are errors $\varepsilon_y>0$ for $y\in Y$ so that
$\sum_{y\in Y}\varepsilon_y<1$ and for all
$x\in X$, $y\in Y$:
\begin{equation}
\arraycolsep.2em
\label{sym}
\begin{array}{rcll}
p(y|x)& = & \varepsilon_y >0\quad
& \hbox{if }x\ne y,\\
p(y|x)& = & \nu_y >\varepsilon_y
& \hbox{if }x=y,\\
\end{array}
\end{equation}
where $\nu_y=1-\sum_{z\ne y}\varepsilon_z$ and thus for all $y$
\begin{equation}
\label{numinuseps}
\nu_y-\varepsilon_y=1-\sum_{z\in Y}\varepsilon_z>0\,.
\end{equation}
\end{definition}

Clearly, every binary channel is input symmetric.  
The matrix in (\ref{insym}) shows an example of an input
symmetric channel with three symbols.
\begin{equation}
\renewcommand{\arraystretch}{1.5}
\label{insym}
\begin{tabular}{|r|lll|}
\hline
& &  \hfil $y$ & \\[-1ex]
\raise 2ex\hbox{$p(y|x)$} & \hfil 0 & \hfil 1 & \hfil 2 \\
\hline
0 & 0.3 & 0.2 & 0.5\\
~$x$ \hfill 1 & 0.1 & 0.4 & 0.5 \\
2 & 0.1 & 0.2 & 0.7 \\
\hline
\end{tabular}
\end{equation} 
By (\ref{numinuseps}), the transition matrix of an input
symmetric channel is the sum of a matrix where each row is
identical (given by the errors) plus $\bigl({1-\sum_{z\in
Y}\varepsilon_z}\bigr)$ times the identity matrix.
Definition~\ref{d-insym} is chosen for our needs and, to
our knowledge, not common in information theory;
the definition of a symmetric channel by Cover and Thomas
(1991, p.~190) is different, but covers the case where
$\varepsilon_y=\varepsilon$ for all $y$.

A channel that is ``output symmetric'' is shown in
(\ref{ternary}), where for any given input $x$ the outputs
$y$ other than $x$ have the same error probabilities
$p(y|x)$.
As we have shown with that example, such a channel may have
codes that are not Nash codes.

The argument for Theorem~\ref{t-bin} below rests on two lemmas.
It is useful to partially order channel outputs and inputs
by ``closeness'' to a given codeword as follows.

\begin{definition}
\label{d-close}
Let $x,y,z\in S^n$ for some set $S$.
Then $y$ is {\em closer to $x$ than $z$} if and only
if~\footnote{ We thank a referee for correcting this
definition.} 
\[
y_j\ne z_j
~~\Rightarrow~~
y_j=x_j
\qquad
\forall j=1,\ldots,n.
\] 
\end{definition}

The following key lemma states in (\ref{decomon}) that the
decoding probability of a channel output $y$ for a state $i$
does not decrease when $y$ gets closer to the codeword $x^i$.

\begin{lemma}
\label{l-out}
Consider a code for an input symmetric channel, a state~$i$,
channel outputs $y$ and $\hat y$, and assume $y$ is closer
to codeword $x^i$ than $\hat y$.
Then
\begin{equation}
\label{inc}
i\in T(\hat y)
~~\Rightarrow~~ 
i\in T(y)\,,
\end{equation}
\begin{equation}
\label{ties}
i\in T(\hat y)
~~\Rightarrow~~ 
T(y)\subseteq T(\hat y)\,,
\end{equation}
and if the code is monotonically decoded then
\begin{equation}
\label{decomon}
d(y,i)\ge d(\hat y,i).
\end{equation} 
\end{lemma}

\proof
To prove (\ref{inc}), we can assume that $y$ and $\hat y$
differ in only one symbol, because then (\ref{inc}) holds in
general via a sequence of changes of only one symbol at a time.
Assume that $y$ and $\hat y$ differ in the $j$th symbol, that
is, $y_j\ne \hat y_j$ and $y_{-j}=\hat y_{-j}$ with the
notation 
\begin{equation}
\label{y-j}
y_{-j}=(y_s)_{s\ne j},
\quad
y=(y_j,y_{-j}).
\end{equation} 
With (\ref{pyx}), we use the notation
\begin{equation}
\label{x-j}
p(y|x)= p(y_j|x_j)\, p(y_{-j}|x_{-j}):=
p(y_j|x_j) \prod_{s\ne j} p(y_s|x_s), 
\end{equation} 
and, for any $k$ in $\Omega$,
\begin{equation}
\label{Qk}
Q_k:= q_kV_k\,p(y_{-j}|x^k_{-j}).
\end{equation} 
Then by (\ref{Yi}),
$y\in Y_i$ means
$q_i V_i\,p(y|x^i) \ge q_k V_k\, p(y|x^k)$ for all $k$ in
$\Omega$, or equivalently
\begin{equation}
\label{qiVi}
q_i V_i\,
p(y_j|x^i_j)\,p(y_{-j}|x^i_{-j})
\ge
q_k V_k\,
p(y_j|x^k_j)\,p(y_{-j}|x^k_{-j}),
\end{equation}
that is, by (\ref{Qk}), $y\in Y_i$ if and only if
\begin{equation}
\label{inYi}
\frac
{p(y_j|x^i_j)}
{p(y_j|x^k_j)}
\ge
\frac
{Q_k}
{Q_i}
\qquad
\forall k\in\Omega~.
\end{equation}
Because $y$ is closer to $x^i$ than $\hat y$, we have
$y_j=x^i_j\ne\hat y_j$.
Suppose, to show (\ref{inc}), that $\hat y \in Y_i$, that
is, because $y_{-j}=\hat y_{-j}$, 
\begin{equation}
\label{yhat}
\frac
{p(\hat y_j|x^i_j)}
{p(\hat y_j|x^k_j)}
\ge
\frac
{Q_k}
{Q_i}
\qquad
\forall k\in\Omega~,
\end{equation}
and we want to show (\ref{inYi}).
For those $k$ where $x^i_j=x^k_j$, the left-hand side of
(\ref{yhat}) does not depend on $\hat y_j$ (and thus holds with
$y_j$ instead of $\hat y_j$), so consider any state $k$
where $x^i_j\ne x^k_j$.
Then by (\ref{sym}),
\begin{equation}
\label{incsym}
\frac
{p(y_j|x^i_j)}
{p(y_j|x^k_j)}
=
\frac{\nu_{y_j}}{\varepsilon_{y_j}}
> 1=
\frac{\varepsilon_{\hat y_j}}{\varepsilon_{\hat y_j}}
\ge
\frac
{p(\hat y_j|x^i_j)}
{p(\hat y_j|x^k_j)}
\ge
\frac
{Q_k}
{Q_i}~\,
\end{equation}
which shows (\ref{inYi}).
So $\hat y\in Y_i$ implies $y\in Y_i$, which proves
(\ref{inc}).  

To show (\ref{ties}), assume again that $y$ and $\hat y$
differ only in their $j$th symbol, and let
$i\in T(\hat y)$ and $l\in T(y)$ for a state~$l$.
That is, $\hat y\in Y_i$ and $y\in Y_l$, where $y\in Y_i$ by
(\ref{inc}).
Then states $i$ and $l$ are tied for~$y$, and clearly
\begin{equation}
\label{i=l}
\frac
{p(y_j|x^i_j)}
{p(y_j|x^l_j)}
=
\frac
{Q_l}
{Q_i}\,.
\end{equation}
If $x^i_j=x^l_j$ then (\ref{i=l}) implies $Q_l=Q_i$
and (\ref{yhat}) holds with $l$ instead of~$i$, so
$\hat y\in Y_l$, that is, $l\in T(\hat y)$.
If $x^i_j\ne x^l_j$, then the strict inequality
(\ref{incsym}) for $k=l$ contradicts (\ref{i=l}), so this
cannot be the case.
This shows (\ref{ties}).

To show (\ref{decomon}), assume monotonic decoding as in
(\ref{mono}).
If $\hat y\not\in Y_i$, then trivially
$d(y,i)\ge d(\hat y,i)=0$.
Otherwise,
$i\in T(\hat y)$ and thus $i\in T(y)\subseteq T(\hat y)$
by (\ref{inc}) and (\ref{ties}), which shows (\ref{decomon})
by~(\ref{mono}).
\endproof

The next lemma\footnote{ We are grateful to a referee who
suggested this step for the binary channel.}
compares two channel inputs $x$ and $\hat x$ that differ in
a single position $j$, and the corresponding channel output when
that $j$th symbol arrives as $y_j$, for arbitrary other
output symbols $y_{-j}$, using the notation~(\ref{y-j}).

\begin{lemma}
\label{l-in}
Consider a monotonically decoded code for an input symmetric
channel, and channel inputs $x$ and $\hat x$ which differ
only in the $j$th symbol, where $x$ is closer to codeword
$x^i$ than $\hat x$.
Then for all $y_{-j}$
\begin{equation}
\label{single}
\sum_{y_j\in Y} p((y_j,y_{-j})\,|\, x)~d((y_j,y_{-j}),i)
\ge
\sum_{y_j\in Y} p((y_j,y_{-j})\,|\, \hat x)~d((y_j,y_{-j}),i).
\end{equation}
\end{lemma}

\proof
Because $x_{-j}=\hat x_{-j}$ and by (\ref{x-j}),
all terms in (\ref{single}) have $p(y_{-j}|x_{-j})$ as a
common factor.  
By taking that factor out and subtracting the right-hand
side, (\ref{single}) is equivalent to
\begin{equation}
\label{position}
\sum_{y_j\in Y} \bigl(p(y_j\,|\, x_j)-p(y_j\,|\, \hat x_j)\bigr)~d((y_j,y_{-j}),i)
\ge0~.
\end{equation}
If $y_j\ne x_j$ and $y_j\ne\hat x_j$, then
$p(y_j\,|\, x_j)-p(y_j\,|\, \hat x_j)
=\varepsilon_{y_j}-\varepsilon_{y_j}=0$, so
(\ref{position}) is equivalent to 
\begin{equation}
\label{onlytwo}
\bigl(p(x_j\,|\, x_j)-p(x_j\,|\, \hat x_j)\bigr)~d((x_j,y_{-j}),i)
+
\bigl(p(\hat x_j\,|\, x_j)-p(\hat x_j\,|\, \hat x_j)\bigr)
~d((\hat x_j,y_{-j}),i)
\ge0~.
\end{equation}
By (\ref{numinuseps}), $p(x_j\,|\, x_j)-p(x_j\,|\, \hat x_j)=
\nu_{x_j}-\varepsilon_{x_j}=
1-\sum_{z\in Y}\varepsilon_z=
\nu_{\hat x_j}-\varepsilon_{\hat x_j}$,
so that (\ref{onlytwo}) is equivalent to 
\begin{equation}
\Bigl(1-\sum_{z\in Y}\varepsilon_z\Bigr)
\Bigl(d((x_j,y_{-j}),i)-d((\hat x_j,y_{-j}),i) \Bigr)
\ge0~,
\end{equation}
which is true because 
$d((x_j,y_{-j}),i)=d((x_j^i,y_{-j}),i)\ge
d((\hat x_j,y_{-j}),i)$ by~(\ref{decomon}).
This shows (\ref{single}).
\endproof

The following main theorem is essentially a corollary to
Lemma~\ref{l-in}.

\begin{theorem}
\label{t-bin}
Every monotonically decoded code for an input symmetric
channel is a Nash code.
\end{theorem}

\proof
For any position $j$, a channel output $y$ is of the form
$(y_j,y_{-j})$ as considered in (\ref{single}).
If $x$ and $\hat x$ differ only in the $j$th position and $x$ is
closer to $x^i$ than $\hat x$, with $x_j=x^i_j\ne\hat x_j$,
then summing (\ref{single}) over all $y_{-j}$ shows 
\[
\sum_{y\in Y^n}
p(y|x)\, d(y,i)\ge
\sum_{y\in Y^n}
p(y|\hat x)\, d(y,i)\,.
\]
For an arbitrary channel input $\hat x$, considering one
symbol at a time where $\hat x$ differs from $x^i$, this
eventually gives (\ref{sender}), which proves the claim.
\endproof

In (\ref{inYi}), it is used that all transition
probabilities of the channel are positive.
In fact, Theorem~\ref{t-bin} does not hold without this
assumption.

\begin{remark}
\label{r-zero}
If some error probabilities are zero, it is no longer true
that every monotonically decoded binary code is a Nash code.
\end{remark}

\proof
Consider a binary ``Z-channel'' where
$p(1|0)=\varepsilon_0=0$ and
$p(0|1)=\varepsilon_1=\varepsilon>0$, which is used twice
($n=2$), with transmission probabilities shown in
(\ref{zchannel}).  
\begin{equation}
\renewcommand{\arraystretch}{1.5}
\label{zchannel}
\begin{array}{c|r|cccc|}
\cline{2-6}
&&&  \multicolumn{2}{c}{~~~y}& \\[-1ex]
\raise 2ex\hbox{$q_iV_i$~} &
\raise 2ex\hbox{$p(y|x)$} &
~~00~~ & 01 & 10 & 11\\
\cline{2-6}
1& 00 &\fbox1 & 0 & 0 & 0\\ 
1&01 & \varepsilon& \fbox{$1-\varepsilon$}&  \fbox0 & \fbox0  \\
& ~\raise2.5ex\hbox{\smash{$x$}}\hfill 10 
& \varepsilon& 0& 1-\varepsilon& 0   \\ 
& 11 & \varepsilon^2 & \varepsilon(1-\varepsilon) &(1-\varepsilon)\varepsilon & (1-\varepsilon)^2  \\
\cline{2-6}
\end{array}
\end{equation} 
Assume uniform weights $q_iV_i=1$ and let the two codewords
be $x^0=00$ and $x^1=01$, so that $Y_0=\{00,10,11\}$ and
$Y_1=\{01,10,11\}$.
Note that outputs 10 and 11 are both tied because they have
probability zero with these inputs.
Assume that these two ``unobtainable'' outputs are decoded as
state 1, which defines a monotonic decoding rule (for a smaller
set of tied states, the probability of decoding a state
in the smaller set does not go down).
This decoding is indicated by boxes in~(\ref{zchannel}).
However, this is not a Nash code because the sender can
improve the probability of decoding state 1 from
$1-\varepsilon$ to $1-\varepsilon^2$ by choosing $\hat x=11$
instead of $x^0=01$ as channel input.
\endproof

\section{Nash-stable channels}
\label{s-stable}

In this section we carry the analysis of
Section~\ref{s-insy} one step further.
This is motivated by Lemma~\ref{l-in} which asserts, in
effect, that the Nash property applies when varying only the
$j$th symbol in the transmitted $n$-tuple.
That is, if a single use of the channel always gives a Nash
equilibrium under monotonic decoding, then this also holds
when the channel is used $n$ times independently, with
codewords of length~$n$.
In fact, each of the $n$ times one can use a different
channel.
We first give a formal statement and proof of this
observation.
Afterwards, we discuss its relationship to the results of
the previous section.

\begin{definition}
\label{d-stable}
A discrete noisy channel
is called {\em Nash-stable} if, for a single use of the
channel ($n=1$), every monotonically decoded code is a Nash
code, for any number of states $i$ with nonnegative weights $q_iV_i$.
\end{definition}

The following theorem considers a {\em product} of $n$ 
noisy channels with input and output alphabets $X{(j)}$
and $Y{(j)}$
and transition probabilities $p_j(y_j|x_j)$ for $1\le j\le n$.
These channels are used independently with channel inputs 
$x=(x_1,\ldots,x_n)$ and channel outputs
$y=(y_1,\ldots,y_n)$, where $y$ is obtained, analogous to
$(\ref{pyx})$, according to
\begin{equation}
\label{pjyx}
p(y|x)=\prod_{j=1}^n p_j(y_j|x_j).
\end{equation}
Note that the possible inputs $x$ to the product channel
have their $n$ symbols distorted with independent errors,
but the considered codes need not have any product
structure.
That is, the codewords can be chosen in any way just as in
the previously considered case of using the same channel $n$
times.

\begin{theorem}
\label{t-gen}
The product of Nash-stable channels is Nash-stable.  
\end{theorem}

\proof
Let  $X=\prod_{j=1}^n X{(j)}$ and $Y=\prod_{j=1}^n Y{(j)}$.
Consider a finite set $\Omega$ of states and a code
$c:\Omega\to X$, where we denote the codewords by $x^i=c(i)$
as usual for $i$ in~$\Omega$.
Assume that the decoding function $d:Y\times\Omega\to\reals$
is monotonic.
If $c$ is not a Nash code, then there is some
state $i$ and $x=x^i$ and $\hat x$ in $X$ so that 
\begin{equation}
\label{lessxhat}
\sum_{y\in Y} p(y|x)\, d(y,i)< \sum_{y\in Y} p(y|\hat x)\, d(y,i)\,.
\end{equation}
As in Theorem~\ref{t-bin}, this implies that
(\ref{lessxhat}) holds for some $x$ and $\hat x$ in $X$ that
differ only in their $j$th symbol with $x$ closer to $x^i$
than $\hat x$, that is, $x_j=x^i_j\ne\hat x_j$, and
otherwise $x_s=\hat x_s$ for $s\ne j$, so we consider this
case.
Analogously to (\ref{x-j}), we write
$p(y|x)= p_j(y_j|x_j)\, p(y_{-j}|x_{-j})$,
and in addition let $Y_{-j}=\prod_{s\ne j}Y(s)$.
Because $x_{-j}=\hat x_{-j}$, (\ref{lessxhat}) is equivalent to 
\[
\sum_{y_{-j} \in Y_{-j}}
p(y_{-j}|x_{-j})
\sum_{y_j\in Y{(j)}}
(p_j(y_{j}|x^i_{j})-p_j(y_{j}|\hat x_{j}))
d((y_j,y_{-j}),i)< 0.
\]
Hence, for at least one $y_{-j}$ we have
$p(y_{-j}|x_{-j})>0$ and
\begin{equation} 
\label{jdev}
\sum_{y_j\in Y{(j)}}
p_j(y_{j}|x^i_{j})
d((y_j,y_{-j}),i)< 
\sum_{y_j\in Y{(j)}}
p_j(y_{j}|\hat x_{j}))
d((y_j,y_{-j}),i).
\end{equation}
(Apart from the notation $Y(j)$ for the output set of the
$j$th channel, this just states that (\ref{position}) does
not hold.)
We claim that (\ref{jdev}) violates the assumption that the
$j$th channel is Nash-stable.
Namely, consider the same set of states $\Omega$
and the code $C:\Omega\to X(j)$ that encodes
state $i$ as $C(i)=x^i_j$.
The original full codeword $x^i= (x^i_j,x^i_{-j})$ is sent across
the product channel~$X$, and the $j$th output symbol $y_j$
is decoded according to $D:Y{(j)}\times\Omega\to\reals$
defined by 
\begin{equation}
\label{D}
D(y_j,i)=d((y_j,y_{-j}),i)
\end{equation}
for the fixed other outputs $y_{-j}$.
We want that this reflects the original best-response
decoding, which requires that the weights $q_iV_i$ are
replaced by $q_iV_i\,p(y_{-j}|x^i_{-j})$ (which are exactly
the weights $Q_i$ in (\ref{Qk})).
Then we obtain the following division of $Y(j)$
into best-response sets $Y_i(j)$, analogous to (\ref{Yi}):
\begin{equation}
\label{Yij}
Y_i(j)=\{y_j\in Y(j)\mid
q_iV_ip(y_{-j}|x^i_{-j})\,p(y_j|x^i_j)\ge 
q_kV_kp(y_{-j}|x^k_{-j})\,p(y_j|x^k_j)
~~\forall k\in\Omega\}.
\end{equation}
Hence, $y_j\in Y_i(j)$ if and only if $(y_j,y_{-j})\in Y_i$,
which shows that $D$ in (\ref{D}) is indeed a best-response
decoding of the single-channel outputs~$y_j$.
Because $d$ is monotonic, so is $D$, because the tied
states~$l$ for $y_j$ (where $y_j\in Y_l(j))$) are those that
are tied for $y=(y_j,y_{-j})$ (where $y\in Y_l$).
Because of (\ref{jdev}), $(C,D)$ is not a Nash equilibrium
and the $j$th channel is not Nash-stable as claimed.
So $c$ is a Nash code for the product channel.
\endproof

Theorem~\ref{t-bin} states that for an input symmetric
channel that is used $n$ times independently, every code is
a Nash code.
In particular, it is a Nash code for $n=1$, so an input
symmetric channel is Nash-stable.
In addition, Theorem~\ref{t-gen} is more general by 
allowing a different channel for each of the transmitted $n$
symbols, but it is straightforward to extend the proof of
Theorem~\ref{t-bin} to this case if each channel is input
symmetric.

The condition of Nash-stability raises a number of
questions.
First, as the proof of Theorem~\ref{t-gen} shows, a large
number of states $i$ might be encoded with input symbols
$x^i_j$ for the $j$th channel, with different weights~$Q_i$,
in order to use the assumption that the $j$th channel is
Nash-stable.
Does it matter if some of these weights $Q_i$ are zero?
They are given by $Q_i=q_iV_i\,p(y_{-j}|x^i_{-j})$, so this
happens when some channel error probabilities are zero.
This case is not excluded in the definition of
Nash-stability or in Theorem~\ref{t-gen}.
However, such channels, for example the binary Z-channel,
are not Nash-stable (which explains Remark~\ref{r-zero}),
according to the following proposition.
We do not consider the trivial case that $p(y|x)=0$ for all
input symbols $x$, when the output symbol $y$ can be omitted
altogether.

\begin{proposition}
\label{r-zchan}
Consider a discrete noisy channel where for some input
symbols $x$ and $\hat x$ and output symbol $y$ we
have $p(y|x)=0$ and $p(y|\hat x)>0$.
Then this channel is not Nash-stable.
\end{proposition}

\proof
Consider $\Omega=\{0,1\}$, $q_0=q_1=1/2$, $V_0=2$,
$V_1=1$, and the code $x^0=x^1=x$, so both states are mapped
to the same channel input~$x$ which cannot be received
as channel output~$y$.
(This example can in fact be obtained from the proofs of
Theorem~\ref{t-gen} and Remark~\ref{r-zero}.)
All outputs $y'$ with $p(y'|x)>0$ are decoded as the state~0
with higher weight.
For the channel output $y$, both states are tied because
this event has probability zero, so $y\in Y_0$.
The receiver can therefore choose $d(y,1)=1$, that is,
decode output $y$ as state~1, and decode all other outputs
$\hat y$ so that $p(\hat y|x)=0$ as state~1 as well.
This decoding is monotonic (the only sets of tied states are 
$\{0,1\}$ and $\{0\}$).
Then in state 1, the sender can change from $x^1=x$ to $\hat
x$ and increase the decoding probability from zero to at
least $p(y|\hat x)$.
This improves her payoff, so the code is not a Nash code.
\endproof

The preceding remark shows that Nash-stability requires 
looking at ``ambiguous'' codes that map more than one state
to the same codeword.
However, it also shows that if all channel transmission
probabilites are positive, then among any states mapped to
the same channel input, only those with maximum weight can
be decoded with positive probability.
Clearly (as argued before in the proof of Proposition~\ref{p-diag}),
``undecoded'' states $i$ so that $d(y,i)=0$ for all~$y$ 
can be ignored when checking Nash-stability.
However, according to Definition~\ref{d-stable}, this still
requires checking many conditions for the possible codes,
weights, and monotonic decoding functions.

It can be shown, but is beyond the scope of this paper,
that it is possible to restrict this check to
deterministic monotonic decoding functions.
Then no more than $|Y|$ states $i$ have the
property that $d(y,i)>0$ for some $y$ in~$Y$.
For all other states, the Nash property holds trivially.
For the weights for these states, there are only finitely
many combinations of producing ties for any output~$y$.
The following remark illustrates this for a channel that is
not input symmetric.

\begin{remark}
\label{symnotnec}
There are Nash-stable channels that are not products of
input symmetric channels.
\end{remark}

\proof
Consider the following channel with three symbols.
\begin{equation}
\renewcommand{\arraystretch}{1.5}
\label{stillnash}
\begin{tabular}{|r|lll|}
\hline
& &  \hfil $y$ & \\[-1ex]
\raise 2ex\hbox{$p(y|x)$} & \hfil 0 & \hfil 1 & \hfil 2 \\
\hline
0 & $4/7$ & $1/7$ & $2/7$\\
~$x$ \hfill 1 & $2/7$ & $4/7$ & $1/7$ \\
2 & $1/7$ & $2/7$ & $4/7$ \\
\hline
\end{tabular}
\end{equation} 
Consider deterministic monotonic decoding functions,
where at most three states have positive probability of
being decoded.
If there is only one state decoded with positive
probability, then the Nash condition holds trivially,
and for three states it holds by Proposition~\ref{p-diag}.
The symbols $0,1,2$ can be cyclically permuted without
changing the channel, so suppose the code for two states $0$
and $1$ uses codewords $x^0=0$ and $x^1=1$.
The decoding depends on the relative weights $q_iV_i$, so
suppose priors are uniform and $V_0=1$.
Then for $1/4<V_1<2$ we have 
$Y_0=\{0,2\}$ and $Y_1=\{1\}$,
which gives a Nash code.
If $V_1<1/4$ then $Y_1$ is empty and $Y_0=\{0,1,2\}$, which
gives trivially a Nash code, and similarly if $V_1>2$.
If $V_1=1/4$, then $Y_0=\{0,1,2\}$ and $Y_1=\{1\}$,
and the two states are tied for $y=1$.
If output $y=1$ is decoded as state 0, then the Nash
property holds trivially, if as state 1, then sending $x^1$
gives the maximum decoding probability $4/7$, so this is
also a Nash code.

If $V_1=2$, then $Y_0=\{0,2\}$ and $Y_1=\{0,1,2\}$,
so that the two states are tied both for $y=0$ and $y=2$.
By consistency, both outputs $y=0$ and $y=2$ are decoded
either as state 0 or as state~1, which correspond to the
cases already considered and give Nash codes.

Finally, it is not hard to see that any mixed decoding
strategy that is monotonic is a convex combination of the
considered deterministic monotonic decoding functions, which implies
the Nash property as well.
This applies also to many states where more than one state
is mapped to the same input symbol.
\endproof 

The computational difficulty of deciding if a given channel
is Nash-stable is open.
The problem belongs to the complexity class co-NP because it
is is easy to verify that the channel is not Nash-stable, by
providing suitable weights, a code, a monotonic decoding
function, and a profitable deviation.
We envisage two possible answers:
Either one can show that Nash-stable channels require that
multiple ties occur simultaneously, like for input symmetric
channels or in the example (\ref{stillnash}), and check only
codes with few states.
In that case, there may be a polynomial-time algorithm.
Alternatively, the problem whether a channel is Nash-stable
may be co-NP-complete.
We leave this as a topic for future research.

\section{General deterministic monotonic decoding functions}
\label{s-mon}

When is a deterministic decoding function monotonic?
Suppose there is some fixed order on the set of states so
that always the first tied state is chosen according to that
order.
In this final section, we show that this is essentially the only
way to break ties with a deterministic monotonic decoding
function if it is defined for {\em all} sets of
tied states $T$ with up to three states.

Because any monotonic decoding function is consistent
according to (\ref{consistent}), it is useful to consider it
as a function $d:\mathcal T\times\Omega\to\reals$ where 
\begin{equation}
\label{calT}
T\in \mathcal T 
\quad
\Leftrightarrow
\quad
T=T(y)=\{l\in\Omega\mid y\in Y_l\} \quad\hbox{for some }y\in Y
\end{equation}
and
\begin{equation}
\label{general}
d(T,i) := d(y,i)
\quad
\hbox{if } T=T(y)
\end{equation}
which is well defined by (\ref{consistent}).
Whether we write $d(T,i)$ or $d(y,i)$ will be clear from the
context.

Consider again the example (\ref{circular}) with
$d(111,1)=1$ as shown on the left in Figure~\ref{binary}.  
The following decoding function, changed from
(\ref{circular}) so that $101$ is decoded as state~1, is
monotonic,
\begin{equation}
\label{nowNash} 
d(110,1)=1,\quad d(011,2)=1,\quad d(101,1)=1,
\quad d(111,1)=1,
\end{equation}
shown in the right picture in Figure~\ref{binary}.
This is a Nash code because all $y$ in the set $Y_1$, see
(\ref{Y0123}), are decoded as state~1;
whichever $\hat x$ in $Y_1$ the sender decides to transmit
instead of $x^1$, there is one $y$ in $Y_1$ for which
$p(y|\hat x)=\eps^2(1-\eps)$, so that the payoff to the
sender in (\ref{pay}) does not increase by changing from
$x^1$ to~$\hat x$.
\def\P{\prec}

As the right picture in Figure~\ref{binary} shows,
the decoding function in (\ref{nowNash}) can be defined by 
the following condition:
Consider a fixed linear order $\P$ on $\Omega$ (in
this case $0\P 1\P 2\P 3$) so that
\begin{equation}
\label{fixedorder}
d(T,i) = 1
~~~\Leftrightarrow~~~
i\in T
~\hbox{ and }~
\forall k\in T,~k\ne i ~:~
i\P k \,.
\end{equation}
That is, the decoding rule chooses the $\P$-smallest state
$i$ from the set~$T$.
A {\em fixed-order} decoding function $d$ fulfills (\ref{fixedorder})
for some~$\P$.
Such a decoding function is deterministic and clearly monotonic.

We want to show that any deterministic monotonic decoding
function is a fixed-order decoding function. 
We have to make the additional assumption that the
decoding function $d(T,i)$ is {\em general\/} in the sense
that it is defined for {\em any\/} nonempty set $T$ 
(where it suffices to require this at least for all $|T|\le 3$),
not only the sets $T$ in $\mathcal T$ that occur as sets of
tied states for some channel output~$y$ as in~(\ref{calT}).

Without this assumption, we could add to the above example 
another state with codeword $x^4=111$ so that the
``circular'' decoding function in (\ref{circular}) is
monotonic and gives a Nash code, but is clearly not a
fixed-order decoding function.
It is reasonable to require that a decoding function is
defined generally and does not just coincidentally lead to a
Nash code because certain ties do not occur (as argued
above, with the decoding (\ref{circular}) we do not have a
Nash code when ties have to be resolved for $y=111$).

For general decoding functions, the monotonicity condition
(\ref{mono}) translates to the requirement 
that for any $T,\hat T\subseteq \Omega$,
\begin{equation}
\label{genmono}
i\in T\subseteq \hat T
~~\Rightarrow~~
d(T,i)\ge d(\hat T,i).
\end{equation}

\begin{proposition}
\label{p-fixedorder}
Suppose that $d(T,i)$ is deterministic and defined for all
nonempty sets $T$ with $|T|\le 3$ (for example, if $\mathcal
T$ in $(\ref{calT})$ contains all these sets) and fulfills
$(\ref{genmono})$.
Then $d$ is a fixed-order decoding function.
\end{proposition}

\proof
Define the following binary relation $\P$ on $\Omega$:
\[
i\,\P\,k
~~~\Leftrightarrow~~~
d(\{i,k\},i)=1.
\]
Clearly, either $i\,\P\,k$ or $k\,\P\,i$ for any two states
$i,k$.
We claim that $\P$ is transitive, that is,
if $i\,\P\,k$ and $k\,\P\,l$, then $i\,\P\,l$. 
Otherwise, there would be a ``cycle'' of distinct $i,k,l$
with $i\,\P\,k$ and $k\,\P\,l$ and $l\,\P\,i$.
This is symmetric in $i,k,l$, so assume
$d(\{i,k,l\},i)=1$ and therefore
$d(\{i,k,l\},k)=0$ and
$d(\{i,k,l\},l)=0$.
However, with $T=\{i,l\}$ and $\hat T=\{i,k,l\}$ we have
$d(T,i)=0<1=d(\hat T,i)$, which contradicts (\ref{genmono}).

So $\P$ defines a linear order on $\Omega$.
We show that (\ref{fixedorder}) holds, that is, for any
$\hat T$ in $\mathcal T$ the decoded state $i$ (so that
$d(\hat T,i)=1$) is the $\P$-smallest element of $\hat T$.
This holds trivially and by definition if $\hat T$ has at most
two elements, otherwise, if $l\,\P\,i$ for some $l\in \hat T$,
then we obtain with $T=\{i,l\}$ the same contradiction
$d(T,i)=0<1=d(\hat T,i)$ as before.
So the decoded state is chosen according to the fixed order
$\P$ on $\Omega$ as claimed.
\endproof

When the weights $q_iV_i$ for the states $i$ are {\em
generic\/}, then $Y_i$ in (\ref{Yi}) is always a singleton,
so no ties occur and decoding is deterministic.
One can make any weights generic by perturbing them
minimally so that ties are broken uniquely but decoding is
otherwise unaffected.
That is, if $i$ and $k$ are tied for some $y$ because
$q_iV_i\, p(y|x^i)= q_kV_k\, p(y|x^k)$, this tie is broken
in favor of $i$ by slightly increasing~$q_iV_i$, which will 
then always happen whenever $i$ and~$k$ are tied originally. 
This induces a fixed-order decoding, where any linear order
among the states can be chosen.
Thus, Proposition~\ref{p-fixedorder} asserts that general
deterministic monotonic decoding functions are those
obtained by generic perturbation of the weights.

Finally, we observe that the above codebook $000,100,010,001$
with decoding as in (\ref{nowNash}) defines a Nash code (and
if priors are minimally perturbed so that $q_1>q_2>q_3$
there are no ties and decoding is unique), but this code is
not locally optimal as in Theorem~\ref{t-opt}.
Namely, by changing the codeword $100$ to $110$, all
possible channel outputs $y$ differ in at most one bit from
one of the four codewords, which
clearly improves the payoff to the receiver.
So not all binary Nash codes are locally receiver-optimal.

\section*{References}
\addcontentsline{toc}{section}{References} 
\frenchspacing
\parindent=-26pt\advance\leftskip by26pt
\parskip=3pt plus1pt minus1pt
\small
\hskip\parindent
Anshelevich, E., et al. (2008),
The price of stability for network design with fair cost
allocation.
SIAM Journal on Computing 38, 1602--1623.  

Argiento R., R. Pemantle, B. Skyrms, and S. Volkov (2009),
Learning to signal: Analysis of a micro-level reinforcement
model.  Stochastic Processes and their Applications 119,
373--390.

Blume, A., and O. J. Board (2013),
Intentional vagueness.
Er\-kennt\-nis,
DOI 10.1007/s10670-013-9468-x, 45 pages.

Blume, A., O. J. Board, and K. Kawamura (2007),
Noisy talk. Theoretical Economics 2, 395--440.

Cover, T. M., and J. A. Thomas (1991),
Elements of Information Theory.
Wiley, New York.  

Crawford, V., and J. Sobel (1982),
Strategic information transmission.
Econometrica 50, 1431--1451.  


De Jaegher, K., and R. van Rooij (2013),
Game-theoretic pragmatics under conflicting and common
interests.
Er\-kennt\-nis, DOI 10.1007/s10670-013-9465-0, 52 pages.

Gallager, R. G. (1968),
Information Theory and Reliable Communication. 
Wiley, New York.

Glazer, J., and A. Rubinstein (2004),
On optimal rules of persuasion.
Econometrica 72, 1715--1736.

Glazer, J., and A. Rubinstein (2006),
A study in the pragmatics of persuasion: A game theoretical approach.
Theoretical Economics 1, 395--410.

Hern\'andez, P., A. Urbano, and J. E. Vila (2010),
Nash equilibrium and information transmission coding
and decoding rules.
Discussion Papers in Economic Behaviour ERI-CES 09/2010, 
University of Valencia.

Hern\'andez, P., A. Urbano, and J. E. Vila (2012),
Pragmatic languages with universal grammars.
Games and Economic Behavior 76, 738--752.

J\"ager, G., L. Koch-Metzger, and F. Riedel (2011),
Voronoi languages: Equilibria in cheap talk games with
high-dimensional types and few signals.
Games and Economic Behavior 73, 517--537.

Kamenica, E., and M. Gentzkow (2011),
Bayesian persuasion.
American Economic Review 101, 2590--2615.  

Koessler, F. (2001),
Common knowledge and consensus with noisy communication.
Mathematical Social Sciences 42, 139--159.  

Kreps, D. M., and J. Sobel (1994),
Signalling.
In: R. J. Aumann and S. Hart, eds.,
Handbook of Game Theory with Economic Applications, Vol. 2,
Elsevier, Amsterdam, 849--867.  

Lewis, D. (1969),
Convention: A Philosophical Study.
Harvard University Press, Cambridge, MA.  

Lipman, B. (2009),
Why is language vague?
Mimeo, Boston University.

MacKay, D. J. C. (2003),
Information Theory, Inference, and Learning Algorithms.
Cambridge University Press, Cambridge, UK.

MacKenzie, A. B., and L. A. DaSilva (2006),
Game Theory for Wireless Engineers.
Morgan and Claypool.  

Monderer, D., and L. S. Shapley (1996), 
Potential games.
Games and Economic Behavior 14, 124--143.

Myerson, R. B. (1994),
Communication, correlated equilibria and incentive compatibility.
In: R. J. Aumann and S. Hart, eds.,
Handbook of Game Theory with Economic Applications, Vol.~2,
Elsevier, Amsterdam, 827--847.

Nowak, M., and D. Krakauer (1999),
The evolution of language. Proc. Nat. Acad. Sci. USA 96,
8028--8033.

Pawlowitsch, C. (2008),
Why evolution does not always lead to an optimal signaling
system. Games and Economic Behavior 63, 203--226.

Shannon, C. E. (1948), 
A mathematical theory of communication.
Bell System Technical Journal 27, 379--423; 623--656.

Sobel, J. (2012),
Complexity versus conflict in communication.
Proc. 46th Annual Conference on Information Sciences and
Systems (CISS).
DOI 10.1109/CISS.2012.6310777, 6 pages.

Sobel, J. (2013),
Giving and receiving advice.
In: Advances in Economics and Econometrics,
Tenth World Congress of the Econometric Society,
D. Acemoglu, M.  Arellano and E. Dekel (eds.), Cambridge
University Press.

Spence, M. (1973),
Job market signaling.
The Quarterly Journal of Economics 87, 355--374.

Srivastava, V., et al. (2005),
Using game theory to analyze wireless ad hoc networks.
IEEE Communications Surveys and Tutorials 7, Issue 4, 46--56.


Touri, B., and C. Lambort (2013),
Language evolution in a noisy environment.
Proc. American Control Conference (ACC), 
1938--1943.

W\"arneryd, K. (1993),
Cheap talk, coordination and evolutionary stability.
Games and Economic Behavior 5, 532--546.

\end{document}